\documentclass[conference,twocolumn,10pt]{IEEEtran} 

\usepackage{epsfig,cite}
\usepackage{graphicx}
\usepackage{color}
\usepackage{amssymb,amsmath}

\IEEEoverridecommandlockouts

\newtheorem{theorem}{Theorem}

\newtheorem{lemma}{Lemma}

\def\given{\:|\:}

\def\Pr{{\mathrm{Pr}}}

\begin{document}

\title{Scaling Laws for Molecular Communication}

\author{
	\authorblockN{Andrew W. Eckford}
	\authorblockA{Dept. of Electrical Engineering and Computer Science\\
	York University\\
	Toronto, Ontario, Canada M3J 1P3\\
	Email: aeckford@yorku.ca}
	\and
	\authorblockN{Chan-Byoung Chae}
	\authorblockA{School of Integrated Technology\\
	Yonsei University, Korea\\
	Email: cbchae@yonsei.ac.kr}
}

\maketitle

\begin{abstract}
In this paper, we investigate information-theoretic scaling laws, independent from communication strategies, for point-to-point molecular communication, where it sends/receives information-encoded molecules between nanomachines. Since the Shannon capacity for this is still an open problem, we first derive an asymptotic order in a single coordinate, i.e., i) scaling time with constant number of molecules $m$ and ii) scaling molecules with constant time $t$. For a single coordinate case, we show that the asymptotic scaling is logarithmic in either coordinate, i.e., $\Theta(\log t)$ and $\Theta(\log m)$, respectively. We also study asymptotic behavior of scaling in both time and molecules and show that, if molecules and time are proportional to each other, then 
the asymptotic scaling is linear, i.e., $\Theta(t)=\Theta(m)$. \\

\end{abstract}



\section{Introduction}

In molecular communication, a transmitter expresses a message in molecules, which propagate towards the receiver via Brownian motion, or some similar means \cite{hiy05}. Molecular communication is found in biological processes such as signal transduction \cite{ein11,eck13}; it has also been proposed as an enabling technology for nanoscale systems \cite{par09}. For this new paradigm of communication, several papers have tried to address the achievable rates (defined as `bits per symbol') of the communication system under theoretical channel and noise assumptions~\cite{JSAC_10,Kuran_NCN10,Chae_JSAC_13}. The author in~\cite{JSAC_10} evaluated, using a circuit model, the normalized gain and delay of the system. The authors in~\cite{Kuran_NCN10,Chae_JSAC_13} studied extensively the basics of molecular communication via diffusion. In~\cite{Kuran_NCN10}, they investigated a new energy model to understand how much energy is required to transmit messenger molecules and \cite{Chae_JSAC_13} introduced several modulation techniques. The authors in~\cite{Chae_JSAC_13} also compared, by using a simple symmetric channel model, the achievable rates. However, most prior work on molecular communication has focused on proposing and analyzing (practical) transmission strategies with theoretical assumptions to achieve higher achievable rates.

The investigation of fundamental capacity limits of molecular communication is still an open problem in information theory. It consistently, however, attracts the attention from researchers since understanding the fundamental limits may provide practical insights. Calculation of mutual information in molecular communication channel is known to be a hard problem. Say there are $m$ molecules, numbered $\{1,2,\ldots,m\}$,
where the $i$th molecule is released at time $x_i$. This molecule takes $n_i$ seconds to 
propagate to the receiver, and arrives at time $x_i + n_i$. So far this looks like a simple additive noise channel -- but the trick is that the molecule arriving at time
$x_i+n_i$ might not be the $i$th molecule to arrive. If the molecules are indistinguishable, then the
releases and arrivals form an order-statistical distribution, which involves a sum over terms for every possible permutation from inputs to outputs (see, e.g., \cite{bap89,eck07}).
For these reasons, unlike better-known channels, we know very little about the Shannon capacity of molecular communication. 
The state of our ignorance about capacity in this channel is such that
it is not even clear what are the right units in which to measure capacity: bits per second? Bits per molecule? Bits per second per molecule? 

While transmission strategies are now relatively well understood~\cite{Chae_JSAC_13}, knowledge about the information-theoretic performance limits is scarce. An early result on achievable information rates was been reported in~\cite{Eckford_08_arxiv}, which provided an \emph{upper bound} in terms of mutual information. Other notable recent efforts in this direction include \cite{rose13}, which
gave lower bounds by exploiting the symmetry of possible input vectors; and \cite{nak12}, which considered capacity in a simplified discrete-time setting. 
Thus, to better understand molecular communication, in this paper, we investigate asymptotic behaviour of the capacity of molecular communication with respect to the number of time intervals and/or the number of molecules. Related work was conducted in \cite{noel13}, which used dimensional analysis
to permit arbitrary scaling of their model.


The rest of this paper is organized as follows. Section~\ref{sec:model} describes the system model under consideration. Sections~\ref{sec:Single} shows scaling results in a single coordinate, i.e., scaling time with constant molecules and scaling molecules with constant time. Scaling in both time and molecules is shown in Section~\ref{sec:both}. 

\section{Model and notation}
\label{sec:model}

First, a brief word on notation: vectors will be represented with superscripts, e.g., $x$ is a scalar,
while $x^t = [x_1,x_2,\ldots,x_t]$ is a vector. It will be clear from context whether a superscript
represents a vector or a scalar exponent. Generally, random variables will be represented
by capital letters (e.g., $Y$), and particular values of those random variables by lower case letters
(e.g., $y$).

\subsection{Molecular communication model}

We use the standard assumptions for information-theoretic analysis of molecular communication 
\cite{nakano-book}:
\begin{enumerate}
	\item The transmitter is a point source of molecules at the origin, and is the only source of the molecule species of interest;
	\item The receiver is a surface surrounding a connected region of points $\mathcal{P}$, which does not include the origin;
	\item Motions of different molecules are independent and identically distributed (i.i.d.), and molecules do not change species or disappear while propagating;
	\item There is no interaction between the transmitter and any molecule after release; and
	\item The medium is infinite in every direction, with no barrier or obstacle except $\mathcal{P}$.
\end{enumerate}
Some of these assumptions may be physically unrealistic: for example, in signal transduction, the transmitter is a cell, which is not well modelled as a point source. However, these assumptions lend themselves to tractable analysis.

To further simplify our analysis, we restrict ourselves to discrete time: the communication session lasts $t$ time instants, indexed $\{1,2,\ldots,t\}$. Meanwhile, the transmitter has $m > 0$ molecules available. It is important to note that the molecules are {\em indistinguishable from each other}.

The transmitter forms the vector $X^t = [X_1,X_2,\ldots,X_n]$, where $X_i$ represents the number of molecules released at discrete time instant $i$.
The receiver forms the vector $Y^t = [Y_1,Y_2,\ldots,Y_t]$, where $Y_i$ is the number
of molecules that arrive at time $t$, obtained as follows. For a molecule released at time $i$, its first arrival time at the receiver is $i + n$, where 
$n$ is the outcome of 
a random variable with distribution $p_N(n)$, the first arrival time distribution of the Brownian motion. 
Thus, $Y_j$ is the number of molecules such that $i+n=j$, for each possible release time $i$.

Recalling that we restrict ourselves to discrete time, $N$ is supported on $\{0,1,2,\ldots\}$.
We further assume,
as in \cite{eck07}, that molecules are absorbed on arrival at the receiver; this can
be shown to be an information-theoretically ideal assumption \cite{nakano-book}. Thus,
$p_N(n)$ is the only property of Brownian motion we require. 

Finally, we require the following conditions on $p_N(n)$ 
to prove our results:
\begin{itemize}
	\item $p_N(n) = 0$ for all $n < 0$, i.e., the system is causal.
	\item Let $F_N(n) = \sum_{i=0}^n p_N(i)$ 
	represent the cdf of the first arrival time distribution; then
there must exist constants $c > 0$ and $n_0 < \infty$ such that $F_N(n_0) \geq c$.
\end{itemize}
Aside from these, we will put no other conditions on the first arrival time distribution $p_N(n)$,
so that our results can apply as widely as possible.

Since we are interested in scaling with increasing $t$ and $m$, we do {\em not} calculate
information rates in this paper; instead, we deal directly with mutual information $I(X^t;Y^t)$.
Reflecting this, we use the notation $C(t)$ or $C(m)$ to indicate capacity as a function of
either time or molecules, respectively. In either case, capacity is found by maximizing
over the input distribution $p_{X^t}(x^t)$.

\subsection{Scaling notation}

Throughout this paper we use Bachmann-Landau scaling notation.
For nonnegative functions $f(n)$ and $g(n)$:
\begin{equation}
	f(n) = \Omega(g(n)), \nonumber
\end{equation}
signifies that there exist positive constants $a$ and $n^\prime$ such that
$a g(n) \leq f(n)$
for all $n \geq n^\prime$ (i.e., $f(n)$ upper bounds $g(n)$); 
\begin{equation}
	f(n) = O(g(n)), \nonumber
\end{equation}
signifies that there exist positive constants $b$ and $n^\prime$ such that
$f(n) \leq b g(n)$
for all $n \geq n^\prime$ (i.e., $g(n)$ upper bounds $f(n)$); and
\begin{equation}
	f(n) = \Theta(g(n)), \nonumber
\end{equation}
signifies that $f(n) = \Omega(g(n))$ and $f(n) = O(g(n))$
(i.e., 
$g(n)$ is of the same order as $f(n)$).

\section{Scaling in a single coordinate}
\label{sec:Single}

\subsection{Overview of main results in this section}

In this section, we consider the scaling of capacity as a function of time,
where number of molecules is held constant, and vice versa. In both cases,
we show that the asymptotic
scaling is logarithmic in the other coordinate: 
in Theorem \ref{thm:tTheta}, we show that
$\Theta(\log t)$ for constant $m$,
and in Theorem \ref{thm:mTheta}, we show that
$\Omega(\log m)$ for constant $t$.

Our approach is to find an upper bound for 
capacity using a maximum-entropy argument,
and a lower bound for capacity using an example communication system.
The results follow by observing that 
the upper and lower bounds have the same asymptotic order.



%

\subsection{Scaling time with constant molecules}

Assume that the number of molecules $m$ is fixed, and evaluate the capacity as the number of time intervals $t$ increases.

\begin{lemma}
	\label{lem:tOmega}
	For fixed $m$,
	\begin{equation}
		C(t) = \Omega(\log t) .
	\end{equation}
\end{lemma}
\begin{proof}
	The proof is found in Appendix \ref{apx:tOmega}.
\end{proof}

\begin{lemma}
	\label{lem:tO}
	For fixed $m$,
	\begin{equation}
		C(t) = O(\log t) .
	\end{equation}
\end{lemma}
\begin{proof}
	Write mutual information as
	\begin{align}
		I(X^t;Y^t) &= H(X^t) - H(X^t \given Y^t) \\
		&\leq H(X^t) \\
		\label{eqn:lem:tO1}
		& \leq \sum_{i=1}^t H(X_i)
	\end{align}
	where (\ref{eqn:lem:tO1}) follows from the chain rule of entropy and the properties
	of conditional entropy. Moreover,
	\begin{equation}
		H(X_i) \leq \log t . \nonumber
	\end{equation}
	Substituting back into (\ref{eqn:lem:tO1}), we have
	\begin{equation}
		I(X^t;Y^t) \leq m \log t. \nonumber 
	\end{equation}
	Since $m$ is constant (by assumption), and since $C(t) = \max_{p_{X^t}(x^t)} I(X^t;Y^t)$, 
	the lemma follows.
\end{proof}

\begin{theorem}
	\label{thm:tTheta}
	For fixed $m$,
	\begin{equation}
		C(t) = \Theta(\log t). \nonumber
	\end{equation}
\end{theorem}
\begin{proof}
	The theorem follows directly from \emph{Lemmas} \ref{lem:tOmega}-\ref{lem:tO},
	and the definition of $\Theta(\log t)$.
\end{proof}

\subsection{Scaling molecules with constant time}

In this section, we assume that the number of time intervals~$t$ is fixed, and evaluate the mutual information as the number of molecules $m$ increases.

\begin{lemma}
	\label{lem:mOmega}
	For fixed $t$, 
	%
	%
	\begin{equation}
		C(m) = \Omega(\log m) .
	\end{equation}
\end{lemma}
\begin{proof}
	The proof is found in Appendix \ref{apx:mOmega}.
\end{proof}

\begin{lemma}
	\label{lem:mO}
	For fixed $t$,
	\begin{equation}
		C(m) = O(\log m) .
	\end{equation}
\end{lemma}
\begin{proof}
	Note that
	\begin{align}
		\label{eqn:lem:mO1pre}
		I(X^t;Y^t) &\leq H(Y^t) \\
		\label{eqn:lem:mO1}
		&\leq \sum_{i=1}^t H(Y_i) ,
	\end{align}
	where (\ref{eqn:lem:mO1}) follows from the chain rule of entropy and the properties
	of conditional entropy. 
	Further, since there are only $m$ molecules in total,
	\begin{equation}
		\label{eqn:lem:mO2}
		H(Y_i) \leq \log m ,
	\end{equation}
	The remainder follows the proof of Lemma \ref{lem:tO}, exchanging $m$ for $t$.
\end{proof}

\begin{theorem}
	\label{thm:mTheta}
	For fixed $t$,
	\begin{equation}
		C(m) = \Theta(\log m). \nonumber
	\end{equation}
\end{theorem}
\begin{proof}
	The theorem follows directly from \emph{Lemmas} \ref{lem:mOmega}-\ref{lem:mO},
	and the definition of $\Theta(\log m)$.
\end{proof}

\section{Scaling in both time and molecules}
\label{sec:both}

The news from Section \ref{sec:Single} is grim: a simplistic reading of these
results would be that capacity scales logarithmically in both $t$ and $m$. 
However, if $m$ is proportional to $t$, the story changes.
In this section we restrict ourselves to the natural case where the number of molecules
$m$ is upper bounded by $\alpha t$, for some constant $\alpha$.
Our main result is to show that $I(X^t;Y^t) = \Theta(m) = \Theta(t)$.
As many authors have pointed out that molecules $m$ are proportional to energy,
then if $m$ is proportional to $t$, this could mean a power constraint.

Our approach in this section is similar to that in Section \ref{sec:Single}: we give a 
maximum entropy result as the upper bound, and a practical system as the lower bound.


For $0 \leq \lambda \leq 1$, let $\mathcal{H}(\lambda)$ represent the binary
entropy function:
\begin{equation}
	\mathcal{H}(\lambda) = \lambda \log \frac{1}{\lambda} + (1-\lambda) \log \frac{1}{1-\lambda}. \nonumber
\end{equation}
We make use of the well-known result that
\begin{equation}
	\label{eqn:BinomialEntropyBound}
	\log {n \choose k} \leq n \mathcal{H} \left( \frac{k}{n} \right), 
\end{equation}	
and the property that, given $n$ indistinct objects and $k$ distinct bins, 
the number of ways to assign objects to bins is 
\begin{equation}
	\label{eqn:bins}
	{n+k-1 \choose k-1 } .
\end{equation}

\begin{lemma}
	\label{lem:mtOmega}
	For some constant $\alpha > 0$, suppose $m \leq \alpha t$. 
	Then $C(t) = \Omega(t)$ and $C(m) = \Omega(m)$.
\end{lemma}
\begin{proof}
	The proof is found in Appendix \ref{apx:mtOmega}.
\end{proof}

\begin{lemma}
	\label{lem:mtO}
	For some constant $\alpha > 0$, suppose $m \leq \alpha t$. Then
	$C(t) = O(t)$ and $C(m) = O(m)$.
\end{lemma}
\begin{proof}
	For convenience, assume
	$\alpha t$ is an integer;
	we first show that $I(X^t;Y^t) = O(t)$. 
	First, how many ways are there to arrange 
	any $m \leq \alpha t$ molecules in $t$ time slots? 
	This is equivalent to arranging exactly $\alpha t$ indistinct objects in $t+1$ 
	distinct bins: 
	for any such assignment,
	there are $m \leq \alpha t$ objects in the first $t$ bins, representing molecules
	assigned to time slots; and
	$\alpha t - m$ objects in bin $t+1$, representing molecules not sent. 
	From (\ref{eqn:bins}),
	the number of assignments $A$ is given by
	\begin{equation}
		\label{eqn:mtO0}
		A = {t + \alpha t \choose t} .
	\end{equation}
	Moreover,
	\begin{align}
		I(X^t;Y^t) &\leq H(X^t) \\
		&\leq \log A \\
		\label{eqn:mtO1}
		&\leq (t + \alpha t) \mathcal{H} \left( \frac{t}{t+\alpha t} \right) \\
		\label{eqn:mtO2}
		&\leq (1+\alpha)t ,
	\end{align}
	where (\ref{eqn:mtO1}) follows from (\ref{eqn:BinomialEntropyBound}) and (\ref{eqn:mtO0}),
	while (\ref{eqn:mtO2}) follows since $\mathcal{H}(\cdot) \leq 1$.
	Moreover, this expression upper bounds $C(t)$, 
	since it upper bounds the maximum of $I(X^t;Y^t)$.
	Finally, (\ref{eqn:mtO2}) is obviously $O(t)$. Since $m \leq \alpha t$,
	$C(m) = O(m)$ if $C(t) = O(t)$ by 
	the $O(\cdot)$ notation, and the lemma follows.
\end{proof}

\begin{theorem}
	\label{thm:mtTheta}
	For some constant $\alpha > 0$, suppose $m \leq \alpha t$.
	Then $C(t) = \Theta(t)$ and $C(m) = \Theta(m)$.
\end{theorem}
\begin{proof}
	The theorem follows directly from \emph{Lemmas} \ref{lem:mtOmega}-\ref{lem:mtO},
	and the definition of $\Theta(\log m)$.
\end{proof}




%
%

\appendix

\subsection{Proof of Lemma \ref{lem:tOmega}}
\label{apx:tOmega}

Divide the interval $t$ into intervals of length $\tau = \lfloor \sqrt{t} \rfloor$.
The number of such intervals $\ell$ is
\begin{align}
	\ell &= \left\lfloor t/\tau \right\rfloor \\
	& \geq \sqrt{t} -1.
\end{align}
First suppose $m=1$.  To transmit data, we select one of the $\ell$ intervals
(uniformly at random) and release our one molecule during that interval. Then
\begin{align}
	H(X^t) &= \log \ell \\
	&\geq \log \left( \sqrt{t} - 1 \right) .
\end{align}

Since $m = 1$, at most one element of $Y^t$ is equal to 1.
At the receiver, suppose $U$ is formed from $Y^t$ as follows: 
if $y_i = 1$, and $(j-1)\tau + 1 \leq i \leq j\tau$, then $U = j$;
if all $y_i = 0$, then $U = \ell+1$. Further, the receiver decides that
the molecule was transmitted at the beginning of the $U$th interval.
Note that there are $\ell+1$ possible outcomes for $U$,
and an error occurs if and only if the molecule takes longer than $\tau$
time units to arrive. Thus,
the probability of error is
\begin{equation}
	\label{eqn:tOmegaPerr}
	P_e = 1-F_N(\tau) ,
\end{equation}
where $F_{N}$ represents the CDF of the first arrival time.

Using Fano's inequality,
\begin{align}
	H(X^t \given U) 
	&\leq (1-F_N(\tau)) \log \ell + \mathcal{H}(1-F_N(\tau)) \\
	\label{eqn:tOmegaFano}
	&\leq (1-F_N(\tau)) \log \left( \sqrt{t} - 1 \right) + 1 ,
\end{align}
where (\ref{eqn:tOmegaFano}) follows from the fact that $\mathcal{H}(\cdot) \leq 1$.
Thus
\begin{align}
	\nonumber \lefteqn{I(X^t;U)} & \\ 
	&= H(X^t) - H(X^t \given U) \nonumber\\
	\label{eqn:tFinal1}
	&\geq \log \left( \sqrt{t} - 1 \right) - (1 - F_N(\tau)) \log \left( \sqrt{t} - 1\right) - 1\\
	&= F_N(\tau)\log \left( \sqrt{t} - 1\right) - 1 ,
\end{align}
By the capacity definition and the data processing inequality,
\begin{align}
	C(t) &\geq I(X^t;Y^t) \geq I(X^t;U) \\
	&\geq F_N(\tau)\log \left( \sqrt{t} - 1\right) - 1 .
\end{align}
Finally, $\log(\sqrt{t}-1) = \Omega(\log(\sqrt{t})) = \Omega(\log t)$.

Finally, we generalize to $m > 1$: suppose the transmitter releases 
{\em all} the molecules at once, and $U$ gives the time of arrival of the {\em first} arriving molecule. Then 
(\ref{eqn:tOmegaPerr}) becomes
\begin{equation}
	P_e = \left( 1-F_N(\tau) \right)^m, \nonumber
\end{equation}
and (\ref{eqn:tFinal1}) becomes
\begin{align}
	I(X^t;Y^t) &\geq \log \left( \sqrt{t} - 1 \right) - (1 - F_N(\tau))^m \log \left( \sqrt{t} - 1\right) - 1 \nonumber\\
	\nonumber
	&\geq \log \left( \sqrt{t} - 1 \right) - (1 - F_N(\tau)) \log \left( \sqrt{t} - 1\right) - 1, 
\end{align}
which follows since $1-F_N(\sqrt{t}) \leq 1$. The remainder of the derivation is identical.


%
%

\subsection{Proof of Lemma \ref{lem:mOmega}}
\label{apx:mOmega}


%
%
%

In this proof, 
suppose a communication scheme works as follows. 
Let $\mathcal{W} = \{W_1,W_2,\ldots,W_n\}$ represent the signalling alphabet,
where each $W_i$ is an integer number of molecules between 0 and $m$.
We form $X^t$ by setting $X_1 = W$ (where $W \in \mathcal{W}$),
and $X_2 = X_3 = \ldots = X_t = 0$. That is, all molecules are released in the first time instant.
At the receiver, we form $U = \sum_{i=1}^t Y_i$ from $Y^t$.

Let $p = F_N(t)$, and let $q = 1-p$.
Chebyshev's inequality can be rewritten
\begin{equation}
	\Pr\left( |U - pW| < k \sqrt{Wpq} \right) \geq 1-\frac{1}{k^2} \nonumber
\end{equation}
Since $W \leq m$,
\begin{equation}
	\label{eqn:Chebyshev2}
	\Pr\left( |U - pW| < k \sqrt{mpq} \right) \geq 1-\frac{1}{k^2}
\end{equation}
The event under the probability can be rewritten
\begin{equation}
	\label{eqn:YRange}
	pW - k \sqrt{mpq} < U < pW + k \sqrt{mpq} .
\end{equation}

For the elements $\{W_1,W_2,\ldots,W_n\}$ of the signalling alphabet, let
\begin{equation}
	\label{eqn:XSelection}
	W_j = 2jk\sqrt{mq/p}.
\end{equation}
The peak signal is $W_n = m$, so
%
	$m = 2nk\sqrt{mq/p}$
%
and
%
	$n = (1/2k)\sqrt{mp/q}$,
%
rounding in each case to the nearest integer as necessary.

Moreover, suppose the elements of $\mathcal{W}$ are uniformly distributed. Then
\begin{align}
	H(X^t) &= \log n \nonumber \\
	&= \frac{1}{2} \log m + \log \frac{1}{2k}\sqrt{\frac{p}{q}} . \nonumber
\end{align}

Let $D(U)$ represent a decoding function such that $D(U) =~j$ if 
\begin{equation}
	\label{eqn:DecodingRange0}
	p2jk\sqrt{\frac{mq}{p}} - k \sqrt{mpq} < U \leq p2jk\sqrt{\frac{mq}{p}} + k \sqrt{mpq} .
\end{equation}
After some manipulation, (\ref{eqn:DecodingRange0}) becomes
\begin{equation}
	\label{eqn:DecodingRange}
	(2j-1)k\sqrt{mpq} < U \leq (2j+1)k\sqrt{mpq} .
\end{equation}
From (\ref{eqn:Chebyshev2})-(\ref{eqn:XSelection}), the probability
of error using $D(U)$ is at most $1/k^2$. 
By Fano's inequality,
%
%
\begin{equation}
	\label{eqn:Fano}
	H(X^t \given U) \leq \frac{1}{k^2} \log (n-1) + \mathcal{H}\left(\frac{1}{k^2}\right) ,
\end{equation}
where $\mathcal{H}$ is the binary entropy function.
Since $n \geq 1$, $2n \geq n+1$,  
so we can relax the bound in (\ref{eqn:Fano}) slightly to
\begin{align}
	H(X^t \given U) &\leq \frac{1}{k^2} \log (n+1) + \mathcal{H}\left(\frac{1}{k^2}\right) \nonumber \\
	&= \frac{1}{k^2} (1 + \log n) + \mathcal{H}\left(\frac{1}{k^2}\right) \nonumber\\
	&= \frac{1}{2k^2} \log m + \frac{1}{k^2}(1+ \frac{1}{2k}\sqrt{\frac{p}{q}}) 
		+ \mathcal{H}\left(\frac{1}{k^2}\right). \nonumber
\end{align}
Finally,
\begin{align}
	C(m) &\geq I(X^t;Y^t) \geq I(X^t;U) \nonumber \\
	\nonumber
	&\geq \frac{1}{2} \log m + \log \frac{1}{2k}\sqrt{\frac{p}{q}} \\
	& - \frac{1}{2k^2} \log m - \frac{1}{k^2}(1+ \frac{1}{2k}\sqrt{\frac{p}{q}}) 
		- \mathcal{H}\left(\frac{1}{k^2}\right) \nonumber \\
	& = \frac{1}{2} \left( 1 - \frac{1}{k^2} \right) \log m + K , \nonumber
\end{align}
where $K$ is constant in $m$; this is clearly $\Omega(\log m)$.

\subsection{Proof of Lemma \ref{lem:mtOmega}}
\label{apx:mtOmega}

We will start by considering the case of $\alpha = 1$, and generalize
the result afterward. 

Consider the following communication scheme: each time instant, we
release a single molecule with probability $r$, and release no molecule with probability
$(1-r)$. Obviously, $m \leq t$.
As before, the receiver forms $Y^t$ by counting the number of arrivals at time $t$.

To simplify the proof, however,
the receiver will actually observe $W^t$, a processed version of $Y^t$:
\begin{equation}
	w_i = 
	\left\{
		\begin{array}{cl}
			1, & y_i \geq 1 \\
			0, & y_i = 0.
		\end{array}
	\right. \nonumber
\end{equation}
We now determine $\gamma_0 := \Pr(w_i = 0 \given x_i = 0)$ (the notation $:=$ signifies
assignment).
First, molecular releases are i.i.d. by assumption. Second, for each $j > 0$,
a molecule arrives at time $i$ if and only if one was released at time $i-j$,
and its propagation delay was $j$. Thus,
\begin{equation}
	\gamma_0 = \prod_{j=1}^{i-1} \Big( 1 - rp_N(j) \Big) . \nonumber
\end{equation}
For $\gamma_1 := \Pr(w_i = 0 \given x_i = 1)$ ,
\begin{equation}
	 \gamma_1 =  \Big( 1 - p_N(0) \Big)\prod_{j=1}^{i-1} \Big( 1 - rp_N(j) \Big) . \nonumber
\end{equation}
For $w,x \in \{0,1\}$, define
\begin{equation}
	g_i(w \given x) :=
	\left\{
		\begin{array}{cl}
			\gamma_x, & w = 0 \\
			1-\gamma_x, & w = 1 ,
		\end{array}
	\right. \nonumber
\end{equation}
and
\begin{equation}
	g_i(w) := r g_i(w \given 1) + (1-r) g_i(w \given 0) . \nonumber
\end{equation}
It should be clear that $g_i(w \given x) = p_{W_i | X_i}(w \given x)$,
and $g_i(w) = p_{W_i}(w)$ is the corresponding marginal.
Finally, let
\begin{align}
	I(W_i;X_i) &= E \left[ \log \frac{p_{W_i | X_i}(w \given x)}{p_{W_i}(w)} \right] \nonumber\\
	&= E \left[ \log \frac{g_i(w \given x)}{g_i(w)} \right] ,\nonumber
\end{align}
and let $I_0 = \min_i I(W_i;X_i)$.
It is straightforward to show that $I_0 > 0$ so long as $p_N(0) > 0$. 
Then
\begin{align}
	\label{eqn:mtOmega1}
	I(Y^t ; X^t) & \geq I(W^t;X^t) \\
	\label{eqn:mtOmega2}
	&= E \left[ \log \frac{p_{W^t | X^t}(w^t \given x^t)}{p_{W^t}(w^t)} \right] \\
	\label{eqn:mtOmega3}
	&\geq E \left[ \log  \frac{\prod_{i=1}^t g_i (w \given x)}{\prod_{i=1}^t g_i (w)} \right] \\
	\label{eqn:mtOmega4}
	&= \sum_{i=1}^t E \left[ \log \frac{g_i (w \given x)}{g_i (w)} \right] \\
	\label{eqn:mtOmega5}
	&= \sum_{i=1}^t I(W_i;X_i) \\
	\label{eqn:mtOmega5}
	&\geq t I_0 ,
\end{align}
where (\ref{eqn:mtOmega1}) follows from the data processing inequality,
(\ref{eqn:mtOmega2}) follows from the definition of mutual information,
and (\ref{eqn:mtOmega3}) follows from the auxiliary channel lower bound for mutual information
(see \cite{arn06}).
Finally, from the last line, $I(Y^t;X^t) =~\Omega(t)$.


To generalize beyond $\alpha = 1$, clearly if $\alpha > 1$ these arguments still apply, since
$m \leq t < \alpha t$. If $\alpha < 1$, we restrict the input to use only $1/\alpha$ of the
time instants, sending nothing at the remaining times; in this case, the final line
in~(\ref{eqn:mtOmega5}) becomes $I(Y^t ; X^t) \geq \alpha t I_0$, which is still $\Omega(t)$.


%
%
%

\bibliographystyle{ieeetr} 
\bibliography{MolecularInfoTheory,infotheory,ref_Mole_Chae}

\end{document}